\documentclass[review,11pt]{elsarticle} 
\usepackage[utf8]{inputenc} 
\usepackage[T1]{fontenc}    
\usepackage{amssymb}
\usepackage{amsthm}
\usepackage{hyperref}       
\usepackage{url}            
\usepackage{booktabs}       
\usepackage{amsfonts}       
\usepackage{nicefrac}       
\usepackage{microtype}      
\usepackage{lipsum}
\usepackage{amsmath}
\usepackage{amsthm}
\usepackage{subfigure}
\usepackage{lineno}
\usepackage[left=2.5cm,right=2.5cm,top=2.5cm,bottom=2.5cm]{geometry}

\newtheorem{theorem}{Theorem}
\newtheorem{lemma}{Lemma}
\newtheorem{proposition}{Proposition}

\theoremstyle{definition}

\journal{\vspace{1cm}}

\begin{document}
\begin{frontmatter}
\title{Effects of anti-infection behavior on the equilibrium states of an infectious disease}

\author[1]{Andrés David Báez Sánchez}
\ead{adsanchez@utfpr.edu.br}
\author[1]{Nara Bobko\corref{cor1}}
\ead{narabobko@utfpr.edu.br}
\cortext[cor1]{Corresponding author}
\address[1]{Mathematics Department, Federal University of Technology, Av. Sete de Setembro, 3165, 80230-901, Curitiba, Paraná, Brazil.}

\begin{abstract}
We propose a mathematical model to analyze the effects of anti-infection behavior on the equilibrium states of an infectious disease. The anti-infection behavior is incorporated into a classical epidemiological SIR model, by considering the behavior adoption rate across the population as an additional variable. We consider also the effects on the adoption rate produced by the disease evolution, using a dynamic payoff function and an additional differential equation. The equilibrium states of the proposed model have remarkable characteristics: possible coexistence of two locally stable endemic equilibria, the coexistence of locally stable endemic and disease-free equilibria, and even the possibility of a stable continuum of endemic equilibrium points. We show how some of the results obtained may be used to support strategic planning leading to effective control of the disease in the long-term.
\end{abstract}

\end{frontmatter}


\section{Introduction}
\label{Sec:1}

The propagation of an infectious disease can be affected by changes in the population behavior and, at the same time, the population behavior concerning the disease can change due to changes in the perception of the epidemiological situation~\cite{ferrer2015risk, boily2005impact,zaidi2013dramatic}. Most recently, in the context of the COVID-19 pandemic, has been clear the relevant  role played by human behavior on the disease dynamic~\cite{Kraemereabb4218, Chinazzieaba9757, betsch2020how} and also has become evident the changes produced on the population behavior and policymakers due to the increase in the number of infected and death
cases~\cite{Remuzzi2020,cowling2020impact,hsiang2020effect}. 

Even before the COVID-19 emergency, there was a  well-recognized demand for mathematical models of infectious diseases considering aspects of the population behavior~\cite{roberts2015nine, brauer2017mathematical}. 

Many mathematical and computational models for infectious diseases based on SIR models have already considered some type of  anti-infection strategies. Some works have incorporated implicitly the possibility of a dynamic preventive behavior, by considering rates of infection or transmission day may depend on some of the  epidemiological variables S, I, or R~\cite{o1997epidemic, ruan2003dynamical, pathak2010rich, lahrouz2012complete, liu2012infectious, seo2013stability, dubey2015dynamics, wang2015sirs, baez2020equilibria}. For other models considering behavioral features see~\cite{manfredi2013modeling}.

Vaccination, as a form of anti-infection behavior, has been considered assuming that part of the susceptible population goes directly into the removed population or adding additional compartments for partially immune population~\cite{yusuf2012optimal,shi2009effect, alam2019three}. For other models considering vaccination see~\cite{rohani2011modeling, martcheva2015introduction,anderson1992infectious}.

In~\cite{bauch2005imitation} a model for vaccination-related behavior is considered using an additional variable corresponding with the rate of vaccination at birth. This new variable  interacts with the infection dynamics in the SIR model and is affected by a differential equation that depends on the infected population $I$. In the present work, we use a similar idea and introduce a behavioral variable related to the adoption rate across the population of some anti-infection behavior. This variable is incorporated into a classical epidemiological SIR model. The dynamics effects on the adoption rate are introduced using an additional differential equation and a dynamic linear payoff depending on the epidemiological variables.

We focus on the study of equilibrium states as an attempt to understand the long-term characteristic and consequences of the interplay between population behavior and disease dynamics.

The equilibrium states of the proposed model have remarkable characteristics: possible coexistence of two locally stable endemic equilibria, the coexistence of locally stable endemic and disease-free equilibria, and even the possibility of a stable continuum of endemic equilibrium points. We will describe how some of the results obtained may be used to support strategic planning leading to effective control of the infectious disease in the long-term.

The paper is organized as follows. In Section~\ref{Sec:2} we develop the mathematical model and discuss some basic characteristics. In Section~\ref{Sec:3} we discuss the existence and stability of its equilibrium points, which is the main focus of the present work. We will show that the set of equilibrium points of the proposed model, have some remarkable characteristics in the context of epidemiological models: coexistence of two locally stable endemic equilibria, the coexistence of locally stable endemic and disease-free equilibria, and the possibility of a stable continuum of endemic equilibrium points. In Section~\ref{Sec:4} we use some of the results to obtain thresholds for parameters leading to effective long-term control of the epidemic disease. We conclude with some final remarks in Section~\ref{Sec:5} and an Appendix presenting proofs of some of the results established in the paper.

\section{A Mathematical Model for an Infectious Disease with an Anti-Infection behavior}
\label{Sec:2}

Compartment models, and particularly SIR models, have been extensively used for mathematical modeling of infectious diseases~\cite{brauer2017mathematical}. The main idea behind SIR models is to consider a population divided into three disjoint categories or compartments: susceptible individuals, infected individuals, and removed (recovered or deceased) individuals, denoted by $S$, $I$, and $R$ respectively. If $N$ denotes the total population, then we have $N=S+I+R$.

Depending on the modeling approach, the variables $S$, $I$, and $R$  may be considered as the absolute numbers of individuals in each group or as the proportion of individuals relative to the total population. In this work, we consider this latter approach. Therefore, considering the time dependency, we have that $S(t) + I(t) + R(t) = 1$ for all $t$.

Within these considerations, an SIR model with vital dynamics and constant population can be stated as
\begin{equation}\label{modelSIR} 
\begin{split}
\dfrac{dS}{dt} &= \mu-\beta\,S\,I-\mu\,S\\
\dfrac{dI}{dt} &= \beta\,S\,I-\mu\,I-\gamma I\\
\dfrac{dR}{dt} &= \gamma \, I-\mu\,R,
\end{split}
\end{equation}
with $S(0) + I(0) + R(0) = 1$. The positive real numbers $\mu$, $\beta$, and $\gamma$ can be interpreted as birth-mortality rate, infection rate, and recovery rate respectively. The constant population consideration is implicit into the system, since $N(t) = 1$ is the only solution of 
$$\dfrac{dN}{dt} = \dfrac{dS}{dt} + \dfrac{dI}{dt} + \dfrac{dR}{dt} = \mu(1-N)$$
satisfying $N(0) = 1$.  For more details about SIR-type models 
see~\cite{rohani2011modeling, martcheva2015introduction}.

Now, consider that there is some behavior or action that can be taken to avoid or reduce the impact of the infection. This behavior can be interpreted as a vaccination initiative, a preventive hygienic measure, a quarantine restraint, or a combination of similar actions. Let $x$ be the proportion of the population following this anti-infection behavior. 

When the population is considering this behavior or action, the perception of the benefit obtained by following it, may not always be constant. In fact, depending on the epidemiological state, the benefit may vary. For example, in a situation with a small proportion of infected, the benefit of adopting the anti-infection behavior may be considered irrelevant for some part of the population.
On the other hand, in a situation where the majority of the population has no immunity, the benefits may be considered high. To analyze this kind of situation, we propose to consider that there exists a perceived payoff or  benefit obtained from the anti-infection behavior that depends on the epidemiological variables $S$, $I$, and $R$ according to a function $p$ given by
\begin{equation}\label{payoff1}
p(S,I,R)=-a_c+a_I\,I+a_S\,S+a_R\,R,
\end{equation}
where $a_c$, $a_I$, $a_S$, and $a_R$ are positive constants. The constant $a_c$ can be interpreted as the fixed cost of adopting the anti-infection behavior, and the constants $a_I$, $a_S$, and $a_R$  can be interpreted as the behavior-adoption benefit associated with the proportion of infected, susceptible, and removed members of the population, respectively. As we have considered that $S+I+R=1$, we have that
\begin{align*}
-a_c+a_I\,I+a_S\,S+a_R\,R 
&= -a_c+a_I\,I+a_S\,S+a_R\,(1-S-I) \\
&= -(a_c-a_R) + (a_I-a_R)\,I+(a_S-a_R)\,S \\
&= -a_0 + a_1\,I+a_2\,S.
\end{align*}
Therefore, the payoff functions can be simplified to obtain
\begin{equation}\label{payoff2}
p(S,I)=-a_0+a_1\,I+a_2\,S.
\end{equation}

Based on the SIR model~\eqref{modelSIR} and the payoff function~\eqref{payoff2}, we propose the following model considering simultaneously the epidemiological variables $(S,I,R)$ and the behavioral state $x$:
\begin{equation}\label{model}
\begin{split}
\dfrac{dS}{dt} &= \mu-(1-x)\beta\,S\,I-\mu\,S\\
\dfrac{dI}{dt} &= (1-x)\beta\,S\,I-\mu\,I-\gamma I\\
\dfrac{dR}{dt} &= \gamma \, I-\mu\,R\\
\dfrac{dx}{dt} &= x(1-x)(-a_0+a_1\,I+a_2\,S)
\end{split}
\end{equation}
with initial conditions in $[0,1]$, and $N(0)= S(0)+I(0)+R(0)=1$. The three initials equations are essentially the SIR model~\eqref{modelSIR} with a variable infection rate depending on the behavioral variable $x$. If $x=1$, there is no infection at all. If $x=0$, the diseases follow the classical SIR dynamics. The fourth equation may be seen as a logistic equation for $x$ with a growth rate depending on the variables $S$ and $I$ and on the cost/payoff parameters $a_0, a_1, a_2$. Thus, depending on the interplay between these values over time, the adoption rate $x$ may increase or decrease, leading also to a dynamically decreasing or increasing infection rate.  The differential equation for $x$ can also be obtained from the replicator equations in evolutionary game theory (see~\cite{weibull1997evolutionary}), applied to a two-behavior game (follow or not follow the anti-infection behavior) with a symmetric payoff given by $-a_0+a_1\,I+a_2\,S$. 

The main goal of the present work is to study the long-term behavior of model~\eqref{model} in terms of its equilibrium points. To achieve this, we will consider a simplified model obtained by re-scaling some of the parameters. 
Considering
\begin{align}\label{rescaling}
\tau &=t\mu; \quad \widetilde{\beta}=\dfrac{\beta}{\mu};\quad \widetilde{\gamma}=\dfrac{\gamma}{\mu}; \quad \widetilde{a}_0=\dfrac{a_0}{\mu}; \quad\widetilde{a}_1=\dfrac{a_1}{\mu};\quad \widetilde{a}_2=\dfrac{a_2}{\mu} \nonumber\\
k&=1+\dfrac{\gamma}{\mu}=1+\widetilde{\gamma} 
\quad \text{ and } \quad 
R_0=\dfrac{\beta}{\mu+\gamma}=\dfrac{\widetilde{\beta}}{1+\widetilde{\gamma}}=\dfrac{\widetilde{\beta}}{k}. 
\end{align}
and replacing in~\eqref{model}, we obtain
\begin{equation}\label{model3}
 \begin{split}
 \dfrac{dS}{d\tau} &= 1-(1-x)kR_0\,S\,I-\,S\\
 \dfrac{dI}{d\tau} &= (1-x)kR_0\,S\,I-kI\\
 \dfrac{dR}{d\tau} &= (k-1)\, I-R\\
 \dfrac{dx}{d\tau} &= x(1-x)(-\widetilde{a}_0+\widetilde{a}_1\,I+\widetilde{a}_2\,S),
 \end{split}
\end{equation}
with initial conditions in $[0,1]$ and $N(0)= S(0)+I(0)+R(0)=1$.

Note that the parameter $k>1$ and the parameter $R_0$ is also a positive real number. The parameter $R_0$ is called the basic reproduction number and has a fundamental role in the description of the equilibria stability in the classical SIR model~\cite{rohani2011modeling, martcheva2015introduction}. The parameter $R_0$ can be interpreted as the number of cases one case generates, on average, in an uninfected population. It represents a measure of the effectiveness of the infection. We introduce below the term $R_p$, that will be important in the forthcoming analysis of equilibrium points
$$R_p = \dfrac{\widetilde{a}_1 - k\widetilde{a}_2}{\widetilde{a}_1 - k\widetilde{a}_0}.$$

Note that $R_p$ depends both on the payoffs associated with the anti-infection behavior and on the population parameter $k = 1+\tfrac{\gamma}{\mu}$. We will see in Section~\ref{Sec:3} that under the effects of the anti-infection behavior, the constant $R_p$ plays a similar role to the one played by the basic reproduction number $R_0$ in the classical SIR model.

We end this section proving that the variables in~\eqref{model3} properly represent population proportions, in the sense that $S, I, R$ and $x$ belongs to the interval $[0,1]$ for all $t \geqslant 0$, and that $N(\tau) = S(\tau)+I(\tau)+R(\tau)=1$ .

\begin{lemma} \label{lemma_invariant} 
The set $\Omega = \{x \in [0,1], S\geq 0, I\geq 0, R\geq 0 \text{ and } S+I+R=1\}$ is positively invariant under~\eqref{model3}.
\end{lemma}

\begin{proof}
Since $x(\tau) = 1$ and $x(\tau)=0$ are stationary solutions of
$$\dfrac{dx}{d\tau} =x(1-x)(-\widetilde{a}_0+\widetilde{a}_1\,I+\widetilde{a}_2\,S),$$
\noindent
the uniqueness of solutions ensures that $x(\tau)\in [0,1]$ for all $\tau\geqslant 0$, whenever $x(0)~\in~(0,1)$. Furthermore, from~\eqref{model3} we have that $\tfrac{dN}{d\tau} = \mu \, (1-N).$ Since $N(0)=1$, follows that  $S(\tau)+I(\tau)+R(\tau)=N(\tau)=1$ for all $\tau\geqslant 0$.

To prove that $S$, $I$, and $R$ are positives, we analyze the behavior of the solutions with initial conditions at the border of $\mathbb{R}^3_{\geqslant0}$. 
\begin{enumerate}[C{a}se 1.]
  \item If $S(0) = 0$ then $\tfrac{dS}{d\tau}(0) = 1 > 0$, therefore $S$ grows locally. 
  \item If $I(0) = 0$ then $\tfrac{dI}{d\tau}(0) = 0$, therefore $I(\tau)$ will remain non-negative.
  \item If $R(0) = 0$ then $\tfrac{dR}{d\tau}(0) = (k-1) I(0)$. In this case, if $I(0)=0$, then $\tfrac{dR}{d\tau}(0) = 0$, whence $R$ will remain non-negative. 
  On the other hand, if $I(0)> 0$ then $\tfrac{dR}{d\tau}(0) > 0$ since $k>1$. Thus $R$ grows locally.
\end{enumerate}
\end{proof}

\section{Equilibrium States}
\label{Sec:3}

\subsection{Existence}

In this subsection, we determine all the possible equilibrium points of  model~\eqref{model3} and its conditions for existence. The following lemma summarizes the results regarding the six different classes of equilibrium points that can be obtained.

\begin{lemma} \label{lemma_equilibriums} 
Any equilibrium point $P~=~(\bar{S},\bar{I},\bar{R},\bar{x})$ of model~\eqref{model3} satisfies that $\bar{I}~=~\tfrac{1}{k}\left(1-\bar{S}\right)$ and  $\bar{R}~=~\left(1-\tfrac{1}{k}\right)\left(1-\bar{S}\right)$. Thus all equilibrium points are determined by the values of $\bar{S}$ and $\bar{x}$. Furthermore, all the equilibrium points of model~\eqref{model3}  fall into one of the following categories: 
\begin{enumerate}[${P}_1$:]
 \item $\bar{S}=1$ and $\bar{x}=0$;
 \item $\bar{S}=1$ and $\bar{x}=1$;
 \item $\bar{S}=1$ and $\bar{x} \in [0,1]$, s.t. $\widetilde{a}_0=\widetilde{a}_2$;
 \item $\bar{S} = \tfrac{1}{R_0}$ and $\bar{x}=0$, s.t. $R_0 > 1$;
 \item $\bar{S} = \tfrac{1}{R_p}$ and $\bar{x}=1-\tfrac{R_p}{R_0}$, s.t. $R_0 > R_p > 1$ and $\widetilde{a}_1\neq k\widetilde{a}_0$;
 \item $\bar{S} =\tfrac{1}{R_0(1-\bar{x})}$ and $\bar{x} \in  \left(0,\tfrac{R_0-1}{R_0}\right)$, s.t. $ R_0>1$ and $k \widetilde{a}_0=\widetilde{a}_1=k\widetilde{a}_2$.
\end{enumerate}
\end{lemma}

\begin{proof}
The equilibrium points of~\eqref{model3} are the solutions in $\Omega$ of the non-linear system
\begin{equation}\label{system_eq_points}
\begin{split}
1-(1-\bar{x})kR_0\,\bar{S}\,\bar{I}-\,\bar{S} &=0\\
 (1-\bar{x})kR_0\,\bar{S}\,\bar{I}-k\bar{I} &=0\\
(k-1)\, \bar{I}-\bar{R} &=0\\
\bar{x}(1-\bar{x})[-\widetilde{a}_0+\widetilde{a}_1\,\bar{I}+\widetilde{a}_2\,\bar{S}] &=0.
\end{split}
\end{equation}
Note from the first equation that $\bar{S}$ can not be equal to zero. Now, adding the first two equations in~\eqref{system_eq_points}, we obtain that any equilibrium point 
must satisfy $1-\bar{S}=k\bar{I}$. Therefore 
\begin{equation}\label{eqauxI}
\bar{I} = \dfrac{1}{k}\left(1-\bar{S}\right)
\end{equation}
and thus, from third equation in~\eqref{system_eq_points}, follows that
\begin{equation}\label{eqauxR}
\bar{R}= \left(1-\dfrac{1}{k}\right)\left(1-\bar{S}\right).
\end{equation}
Thus, if $\bar{S}=1$, then~\eqref{eqauxI} and~\eqref{eqauxR} implies that $\bar{I}=\bar{R}=0$ and the expressions for equilibrium types ${P}_1, {P}_2$ and ${P}_3$ can be obtain from fourth equation in~\eqref{system_eq_points}.

If $\bar{S}\not= 1$, then~\eqref{eqauxI} implies that $\bar{I}\not=0$. Thus, from second equation in~\eqref{system_eq_points}, we obtain 

$$(1-\bar{x})R_0\bar{S}=1,$$ 
which implies that in this case $\bar{x}\not=1$ and therefore 
\begin{equation}\label{eqauxS}
\bar{S}=\dfrac{1}{R_0(1-\bar{x})}.
\end{equation}
Equation~\eqref{eqauxS} implies the expression for equilibrium $P_4$ but additionally, can be used jointly with equation \eqref{eqauxI} and the fact that $R_p = \tfrac{\widetilde{a}_1 - k\widetilde{a}_2}{\widetilde{a}_1 - k\widetilde{a}_0}$, obtain by basic manipulations of the fourth equation in~\eqref{system_eq_points}, the expressions and conditions defining $P_5$ and $P_6$ .
\end{proof}

\subsection{Comments on Lemma~\ref{lemma_equilibriums}} 
Model~\eqref{model3} has more possible equilibrium points that the classic SIR model. Indeed, the classical SIR model has only two equilibrium points: a disease-free equilibrium and an endemic equilibrium that corresponds precisely to equilibria $P_1$ and $P_4$. In addition,  model~\eqref{model3} have other disease-free equilibria ($P_2$ and $P_3$) and other endemic equilibria ($P_5$ and $P_6$). 

The equilibrium points $P_1$ and $P_2$ differs only in the last component: in $P_1$ no one is adopting the anti-infection behavior and in $P_2$ all population does. Although $P_2$ seems an ideal scenario, it may not be realistic even if the prevention policy has an insignificant cost.

Equilibrium type $P_3$ also differs from $P_1$ only in the last component. However, note that $P_3$ represents an infinite set of equilibrium, since for each $\bar{x}$ we obtain a different equilibrium point. In particular, $P_3$ include $P_1$ and $P_2$ when $\bar{x} = 0$ and $\bar{x} = 1$, respectively. In fact, $P_3$ represents a connected path between these two disease-free equilibria. 

Note that the family of equilibria $P_3$ exists only if $\widetilde{a}_0=\widetilde{a}_2$. In terms of the original parameters, this is equivalent to $a_c = a_S$, that is, the fixed cost has to be exactly equal to the payoff associated with the proportion of susceptible members of the population. Such equality between parameters may be unrealistic, thus we consider $P_3$ of minor practical interest. This also applies to equilibrium family $P_6$ which has also a condition for its existence involving equality between parameters.

As mentioned before, $P_4$ corresponds to the endemic equilibrium of the classical SIR model and has the same existence condition ($R_0 > 1$) in that context.

In turn, the equilibrium point $P_5$ does not coincide with any equilibrium of the classic SIR model and can be considered as a more realistic scenario. In the $P_5$ case, the infection is present ($\bar{I}\neq0$) and only a part of the population adopted the anti-infection behavior. Note also that the condition $R_p < R_0$, implies that the proportion of the susceptible population in $P_5$ is greater than in $P_4$. Consequently, the proportion of infected population in $P_5$ is lower than in $P_4$. Therefore, \textbf{$P_5$ can be interpreted as a desirable  situation where anti-infection behavior reduces the impact of the disease in the long-term}. 

Note also that in this $P_5$ scenario, for a fixed value of $R_0$, the larger is $R_p$, the smaller is the proportion of infected people. This relationship between $R_0$, a parameter related only to the disease, and $R_p$, a parameter related to the cost of intervention, allows an analysis of the effects of behavior and  cost/payoff changes in the disease dynamic. 
The best-case scenario would be one with a minimal value for $\bar{I}$, or equivalently,  a maximal value for $\bar{x}$. 
This will occur if $R_p$ tends to $1$ and in the limit this will imply  $\widetilde{a}_2 = \widetilde{a}_0$ (existence condition of $P_3$). 

The worst-case scenario for $P_5$ would be one where  $R_p$ goes to $R_0$ because in this case, $\bar{x}$ goes to zero and  $P_5$ goes to $P_4$. 

Equilibrium type $P_6$ represents an infinite set of endemics equilibrium points, one for each $\bar{x} \in \left(0,\tfrac{R_0-1}{R_0}\right)$. 
Unlike disease-free equilibria $P_3$, in $P_6$  the value of $\bar{x}$ will affect the value of $\bar{S}$ ($\bar{I}$ and $\bar{R}$ too). Note that if $\bar{x}$ approach  $\tfrac{R_0-1}{R_0}$, then $\bar{S}$ approach  $1$. This means that if the proportion of the population adopting the prevention behavior increase, the proportion of susceptible population also increases (and the proportion of infected population decrease).

Note that, when $\bar{x}$ goes to $0$, $P_6$ goes to $P_4$, and when $\bar{x} = \tfrac{R_0-1}{R_0}$, the equilibrium $P_6$ goes to a $P_3$ equilibrium point. In fact, when  $k \widetilde{a}_0=\widetilde{a}_1=k\widetilde{a}_2$ both sets of equilibria $P_3$ and $P_6$ coexist and have a linking point at $\left(1,0,0,\tfrac{R_0-1}{R_0}\right)$. 
Lastly, note that equilibrium points $P_6$ cannot co-exist with equilibrium point $P_5$, since its existence conditions are incompatible.

\subsection{Jacobian Matrix and its Characteristic Polynomial} 

We are interested in study the stability of equilibrium points of \eqref{model3}. 
Then, it will be useful to consider the associated Jacobian matrix given by:
{\small
$$J(\!S,\!I,\!R,\!x)\!=\! 
\begin{bmatrix} 
 -(1-x)I k R_0 -1 & -(1-x)k R_0 S & 0 & I k R_0 S \\
 (1-x)I k R_0 & (1-x)k R_0S -k & 0 & -I k R_0 S \\
0& k-1 &-1&0\\ 
(1-x) x\widetilde{a}_2 & (1-x) x\widetilde{a}_1 & 0& (1-2 x)(-\widetilde{a}_0+\widetilde{a}_1 I+\widetilde{a}_2 S) \\
\end{bmatrix}.
$$
}
The characteristic polynomial of $J(S,I,R,x)$ can be written as:
\begin{equation}\label{eq_polynomial}
p(\lambda) = |J(S,I,R,x)-\lambda I|= (-1-\lambda) q(\lambda)
\end{equation}
where
{\small
\begin{equation*}
q(\lambda) = \left |
\begin{matrix} 
 -(1-x)I k R_0 -1 -\lambda & -(1-x)k R_0 S  & I k R_0 S \\
 (1-x)I k R_0 & (1-x)k R_0S -k -\lambda & -I k R_0 S \\
(1-x) x\widetilde{a}_2 & (1-x) x\widetilde{a}_1 &  (1-2 x)(-\widetilde{a}_0+\widetilde{a}_1 I+\widetilde{a}_2 S)-\lambda \\
\end{matrix}\right |.
\end{equation*}
}

\subsection{Stability of $P_1$, $P_2$, $P_4$, and $P_5$}

It is clear from \eqref{eq_polynomial} that for any equilibrium point, the Jacobian has at least one negative eigenvalue  $\lambda_1=-1$ and that additional eigenvalues can be studied analyzing the equation $q(\lambda) = 0$. This can be used to establish the following subsection result about the stability of equilibrium points  $P_1$, $P_2$, $P_4$, and $P_5$ whose  complete proof is presented in the Appendix.

\begin{theorem} \label{thm_stability_J} 
Consider system~\eqref{model3}. Assume that $\widetilde{a}_0 \neq \widetilde{a}_2$, $\widetilde{a}_0 \neq \widetilde{a}_1/k$, $R_0 \neq 1$, and $R_0 \neq R_p$. 
\begin{enumerate}[1.]
    \item If $R_0 < 1$ then
    \begin{enumerate}[i.]
        \item $P_1$ is locally asymptotically stable if $\widetilde{a}_0 > \widetilde{a}_2$; \label{thm_stability_J_1i}
        \item $P_2$ is locally asymptotically stable if $\widetilde{a}_0 < \widetilde{a}_2$, \label{thm_stability_J_1ii}
        \item $P_4$ and $P_5$ do not exist. 
    \end{enumerate}
    \item If $R_0 > 1$ and $R_0 < R_p$, then
    \begin{enumerate}[i.]
        \item $P_4$ is locally asymptotically stable if $\widetilde{a}_0 > \widetilde{a}_2$; \label{thm_stability_J_2i}
        \item $P_2$ is locally asymptotically stable if $\widetilde{a}_0 < \widetilde{a}_2$; \label{thm_stability_J_2ii}
    	\item $P_1$ is not stable;
    	\item $P_5$ do not exist. 
    \end{enumerate}
    \item If $R_0 > 1$ and $R_0 > R_p$,
    \begin{enumerate}[i.]
        \item $P_5$ is locally asymptotically stable if $\widetilde{a}_0 > \widetilde{a}_2$, and $\widetilde{a}_0 < \widetilde{a}_1/k$;\label{thm_stability_J_3i}
        \item $P_4$ is locally asymptotically stable if $\widetilde{a}_0 > \widetilde{a}_2$, and $\widetilde{a}_0 > \widetilde{a}_1/k$;\label{thm_stability_J_3ii}
        \item $P_2$ is locally asymptotically stable if $\widetilde{a}_0 < \widetilde{a}_2$, and $\widetilde{a}_0 < \widetilde{a}_1/k$;\label{thm_stability_J_3iv} 
        \item $P_2$ and $P_4$ are locally asymptotically stable if $\widetilde{a}_0 < \widetilde{a}_2$, and $\widetilde{a}_0 > \widetilde{a}_1/k$;\label{thm_stability_J_3iii}
    \item $P_1$ is not stable.
    \end{enumerate}
\end{enumerate}
\end{theorem}

\subsection{Comments on Theorem~\ref{thm_stability_J}} 

In the classic SIR model~\eqref{modelSIR}, when the basic replication rate is sufficiently low ($R_0 < 1$), the disease-free equilibrium point is stable, so the infection does not become an epidemic. As described in Theorem~\ref{thm_stability_J}, this phenomenon also occurs in system~\eqref{model3} but in this case, there are two possible disease-free equilibrium: $P_1$ (zero behavior adoption) and $P_2$ (complete behavior adoption). The values of $\widetilde{a}_0$ and $\widetilde{a}_2$ determine which one is stable.

When the disease is more infectious ($R_0>1$), the classic SIR model admits only one possibility: the endemic equilibrium is stable and the disease-free equilibrium is unstable.
Cases~\eqref{thm_stability_J_2i} and~\eqref{thm_stability_J_3ii} of the Theorem~\ref{thm_stability_J} are equivalent to this situation, since $P_4$ is equivalent to the endemic equilibrium of the classical SIR model. However, in model~\eqref{model3} some more realistic behaviors may occur. Note for example that it is possible that a disease-free equilibrium $P_1$ and the endemic equilibrium $P_4$ coexist simultaneously, both being locally stable (Theorem~\ref{thm_stability_J}~\eqref{thm_stability_J_3iv}). Figure~\ref{fig_sim_1} illustrates this interesting case. Note also that in this situation, equilibrium points $P_1$ and $P_5$ also exist but are not stable.

From Theorem~\ref{thm_stability_J}~\eqref{thm_stability_J_2ii} and~\eqref{thm_stability_J_3ii} another remarkable behavior can be observed, even  if $R_0>1$, it is possible that the system has a disease-free and unique stable equilibrium. Figure~\ref{fig_sim_3} illustrates this situation. In this case, the equilibrium $P_1$ exists and is unstable and equilibrium $P_5$ does not exit. 

\begin{figure}[ht!]
\centering
\subfigure[][Coexistence of Stable Equilibrium $P_2$ (in blue) and $P_4$ (in red).]{\includegraphics[width=0.4\textwidth]{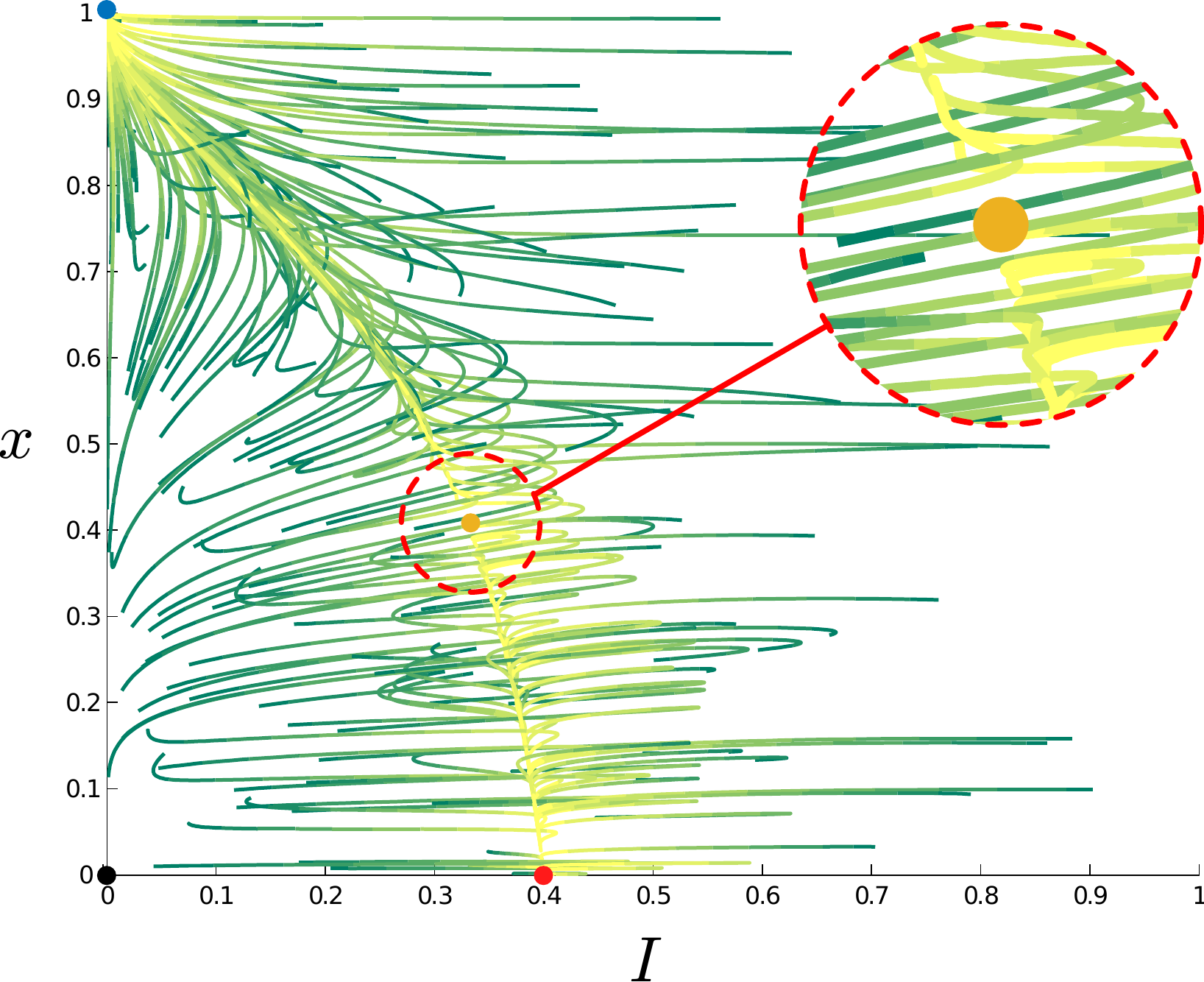}\label{fig_sim_1}}
\subfigure[][$P_2$ (in blue) stable with $R_0 > 1$.]{\includegraphics[width=0.4\textwidth]{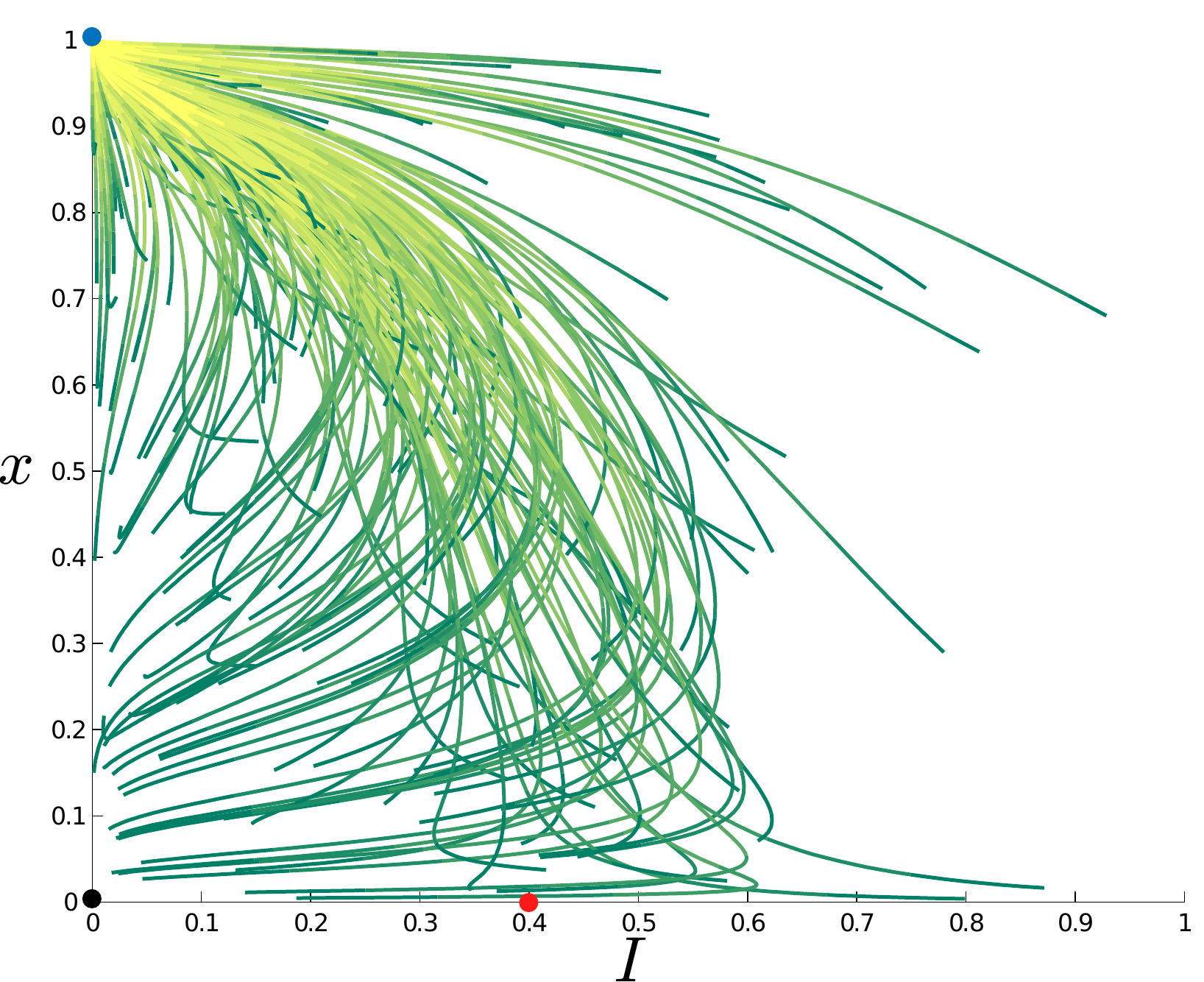}\label{fig_sim_3}}
\subfigure[][Locally Stable Equilibrium $P_5$ (in orange).]{\includegraphics[width=0.4\textwidth]{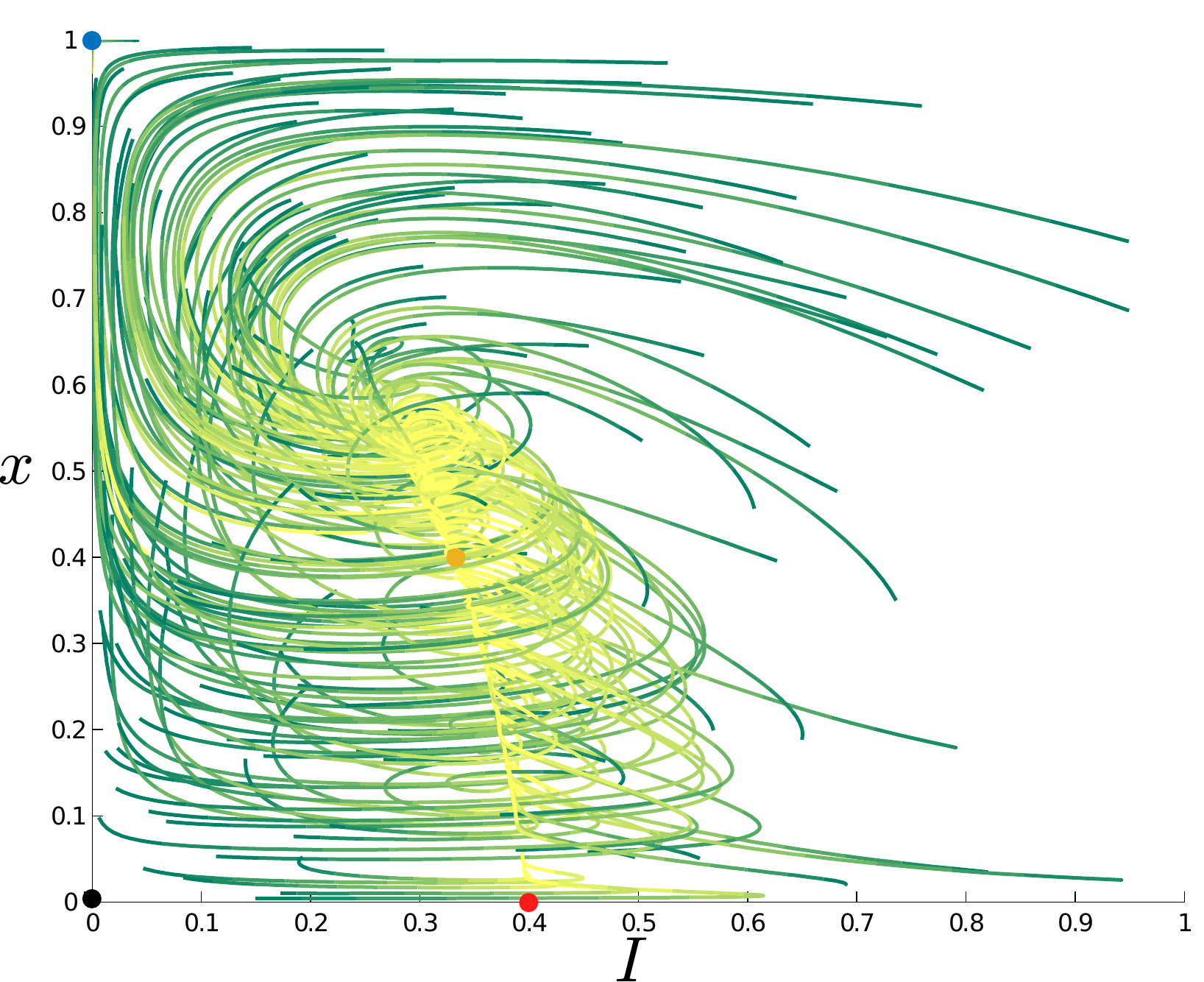}\label{fig_sim_2}}
\subfigure[][Coexistence of the Families $P_3$ (in blue) and $P_6$ (in red).]{\includegraphics[width=0.4\textwidth]{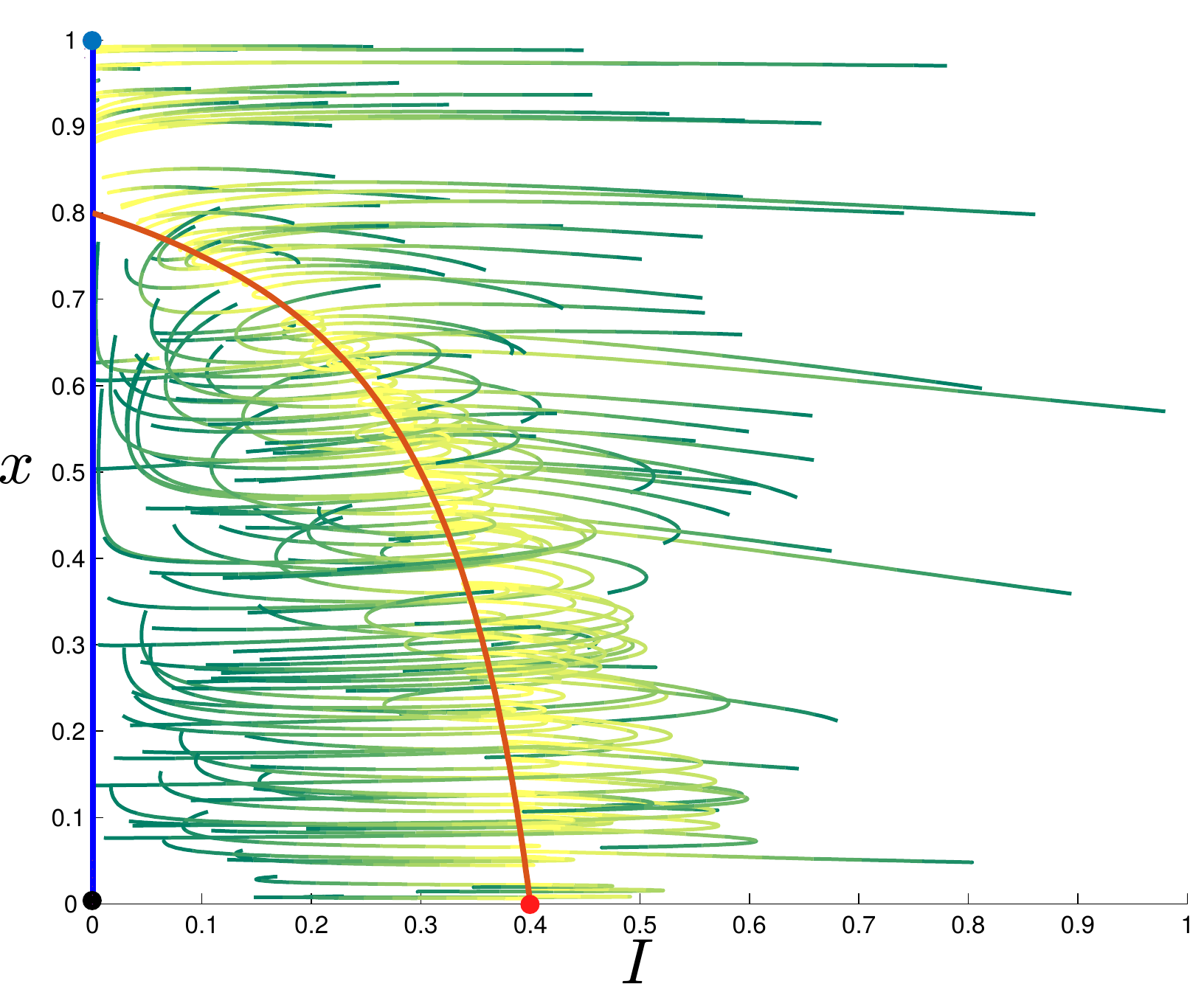}\label{fig_sim_4}}
\caption{\small {\bf Solution curves $I(\tau) \times x(\tau)$ of system~\eqref{model3}.} 
Numerical simulations of solutions with different initial conditions.
Temporal evolution is represented using dark green for initial trajectory points and a gradual variation to yellow as time increase.  
The equilibrium points $P_1$, $P_2$, $P_4$, and $P_5$ are denoted by the dots in the color black, blue, red, and orange, respectively, while the families of equilibria $P_3$ and $P_6$ are denoted by the lines in the color blue and red, respectively (when they exist).
In all cases pictured $k=2$ and $R_0 = 5$.
In~\ref{fig_sim_1} $\widetilde{a}_0=\widetilde{a}_1=1$ and $\widetilde{a}_2=2$. 
Thus $R_p = 3$ and $R_0 > \max\{1,R_p\}$, ensuring that $P_2$ and $P_4$ are both locally stable, while $P_1$ and $P_5$ are unstable (the instability of $P_5$ is highlighted in zoom).
In~\ref{fig_sim_3} $\widetilde{a}_0=1$, $\widetilde{a}_1=7$, and  $\widetilde{a}_2=2$. 
Thus $R_p = 0.6$ and $R_0 > \max\{1,R_p\}$, ensuring that $P_2$ are locally stable, while $P_1$ and $P_4$ are unstable.  
In~\ref{fig_sim_2} $\widetilde{a}_0=3$, $\widetilde{a}_1=7$, and  $\widetilde{a}_2=2$. 
Thus $R_p = 3$,  $R_0 > \max\{1,R_p\}$, and $\widetilde{a}_1/k>\widetilde{a}_0>\widetilde{a}_2$, ensuring that $P_5$ are locally stable, while $P_1$, $P_2$, and $P_4$ are unstable.  
In~\ref{fig_sim_4} $\widetilde{a}_0=\widetilde{a}_2=1$ and  $\widetilde{a}_1=2$. 
Thus, Theorem~\ref{thm_stability_J} ensures that whole family $P_6$ is stable.
}
\label{fig_simulations}
\end{figure}

Finally, besides $P_4$, we have the possibility of another endemic stable equilibrium: the equilibrium point $P_5$.
Assuming that conditions of Theorem~\ref{thm_stability_J}~\eqref{thm_stability_J_3i} are satisfied, equilibria $P_1$, $P_2$, and $P_5$ exists, but only the last one is stable. This equilibrium is particularly interesting because it represents a more favorable epidemiological situation than the equilibrium point $P_4$. Indeed, as $R_p < R_0$, the value of $\bar{I}$ at $P_5$ is smaller than the value of $\bar{I}$ at $P_4$.

\newpage

\subsection{On the stability of equilibria family $P_6$}

As mentioned before, the existence conditions for  equilibrium families $P_3$ and $P_6$ involve equality between some parameters which can be unrealistic. The corresponding stability analysis can not be done  using the  standard approach  based on the Jacobian matrix as in Theorem~\ref{thm_stability_J}, because in this cases, the corresponding Jacobian matrix have a null eigenvalue. In fact, solving equation~\eqref{eq_polynomial} for equilibrium $P_3$ equilibrium lead us to
$$\lambda_1 = -1, \quad \lambda_2=-1, \quad \lambda_3 = 0 \quad \text{ and }\quad \lambda_4 =k (R_0 (1-\bar{x})-1),$$ 
and in a similar fashion, the $P_6$ equilibrium points also have a null eigenvalue~\footnote{Null eigenvalues appears also for $P_1$ if $\widetilde{a}_2=\widetilde{a}_0$ or $R_0 = 1$, and for $P_2$ when $\widetilde{a}_2=\widetilde{a}_0$}.

It can be noted however that some points in the equilibrium family $P_3$ may be locally stable, as illustrated in the Figure~\ref{fig_sim_4}. When $\bar{x} = \tfrac{R_0-1}{R_0}$, then $P_6$ becomes $(1,0,0,\bar{x})$, so this point is a linking point between $P_6$ and $P_3$. Note that this point acts as the threshold between stable and unstable equilibrium points in $P_3$.

Nevertheless, a closer look at system~\eqref{model3} and to the conditions for the existence of $P_6$, allow us to determine some stability conditions for equilibria $P_6$ presented in the following theorem.

\begin{theorem} \label{thm_stability_P6}
Assume that in model~\eqref{model3} we have $R_0>1$ and $k\widetilde{a}_0=\widetilde{a}_1=k\widetilde{a}_2$, so the family of equilibria $P_6$ exists.
If $-\tfrac{R_0k}{(R_0-1)(k-1)}<\widetilde{a}_0$, then the family of equilibria $P_6$ is stable.
\end{theorem}

\begin{proof}
Note first that if $-\tfrac{R_0k}{(R_0-1)(k-1)}<\widetilde{a}_0$  then for all $\bar{x} \in \left(0,\tfrac{R_0-1}{R_0}\right)$ we have that
\begin{equation}\label{ineq1}
-\dfrac{(1-\bar{x})R_0k}{\bar{x}(k-1)}<\widetilde{a}_0,
\end{equation}
because $0<\bar{x}<\tfrac{R_0-1}{R_0}$ implies that $-\tfrac{1-\bar{x}}{\bar{x}}<-\tfrac{1}{R_0-1}$. Now, if  $k\widetilde{a}_0=\widetilde{a}_1=k\widetilde{a}_2$, we have from third and fourth equation of the system~\eqref{model3} that
\begin{equation*}
\begin{split}
\dfrac{dx}{d\tau} &= x(1-x)[-\widetilde{a}_0+k\widetilde{a}_0\,I+\widetilde{a}_0\,S]\\
              &=\widetilde{a}_0 x(1-x)[(k-1)I-R]\\
              &=\widetilde{a}_0 x(1-x)\dfrac{dR}{d\tau}.
\end{split}
\end{equation*}
Thus, we have that 
\begin{equation*}
\dfrac{dx}{dR}=\widetilde{a}_0 x(1-x),
\end{equation*}
and we can  express $x$ in terms of $R$ as 
$$x(R)=\dfrac{e^{\widetilde{a}_0 R}}{e^{\widetilde{a}_0 R}+C_1}.$$
\noindent
This consideration allows us to eliminate the differential equation for $x$ in~\eqref{model3} and using that $S=1-I-R$, we can reduce model~\eqref{model3} to a simplified epidemic model with a recovered-dependant infection described as
\begin{equation} \label{model5}
\begin{split}
\dfrac{dI}{d\tau} &=I[f(R)\,(1-I-R)-k]\\
\dfrac{dR}{d\tau} &= (k-1)\, I-R,\\
\end{split}
\end{equation}
where  $f(R)=(1-x(R))kR_0$.  Recovered-dependent epidemic models as~\eqref{model5} were considered by the authors in~\cite{baez2020equilibria}. In particular,  Theorem 4.3 in~\cite{baez2020equilibria} establish the following result:
\\
\textit{If $f$is positive function, differentiable on $[0,1]$ and $(I^*,R^*)$  is an endemic equilibrium point of~\eqref{model5} such that $\tfrac{df}{dR}(R^*) <\tfrac{1}{k-1}f^2(R^*)$ then $(I^*,R^*)$ is a locally stable equilibrium point.}
\\
Note that if $f(R)=(1-x(R))kR_0$, $f$ is in fact a positive differentiable function on $R$. Additionally, the following inequalities equivalences holds:
\begin{align*}
\dfrac{df}{dR}(R^*)&<\dfrac{1}{k-1}f^2(R^*)\\
-kR_0\dfrac{dx}{dR}&<\dfrac{1}{k-1}f^2(R^*)\\
-kR_0\widetilde{a}_0 \overline{x}(1-\overline{x})&<\dfrac{(1-\overline{x})^2k^2R^2_0}{k-1}\\
-\widetilde{a}_0\overline{x}&<\dfrac{(1-\overline{x})kR_0}{k-1}\\
-\dfrac{(1-\overline{x})R_0k}{\overline{x}(k-1)}&<\widetilde{a}_0,
\end{align*}
which we already showed in \eqref{ineq1} is valid when $-\tfrac{R_0k}{(R_0-1)(k-1)}<\widetilde{a}_0$. Therefore, we conclude that the whole family of equilibria $P_6$ is stable.
\end{proof}

\section{Controlling the infection through population behavior: Choosing the right payoffs}
\label{Sec:4}

In this section, we use the results in Theorem~\ref{thm_stability_J} to find conditions on the behavioral payoffs, that produce a diminishing on the infected population at a stable equilibrium. This can be interpreted as specific policy actions leading to reduce and control the infection in the long-term.
 
According to system~\eqref{model3}, an infectious disease with a small replication rate ($R_0<1$), requires no anti-infection behavior to be eradicated, since the possible stable points $P_1$ and $P_2$ are both disease-free. Nevertheless, the stability conditions in part 1. of Theorem~\ref{thm_stability_J} can be rewritten in terms of the original parameters as follows: if $a_c > a_s$, then $P_1$ is locally asymptotically stable; if $a_c < a_s$, then $P_2$ is locally asymptotically stable. This can be interpreted in terms of public policies, as a quantification of how much reduction on the fixed cost $a_c$ is necessary to achieve full adoption of an anti-infection behavior; \textbf{if $a_c$ is smaller than $a_s$, then in the long-term everyone tends to follow the prevention behavior, even if the disease is poorly infectious ($R_0<1$)}. 

We focus now on the situation when $R_0> 1$ and therefore, the infectious disease may became endemic. We aim to determine, in terms of $R_p$, $a_c$, $a_S$, $a_I$, and $a_R$, successful intervention strategies to control the disease. We consider two scenarios:

\textbf{Scenario 1:} Assume that $a_c < a_S$ and therefore  $\widetilde{a}_0 < \widetilde{a}_2$. In this case, from parts 2. and 3. in Theorem~\ref{thm_stability_J} we have two possibilities: only the disease-free equilibrium $P_2$ is stable (cases~\eqref{thm_stability_J_2ii} and~\eqref{thm_stability_J_3iii}), or $P_2$  and the endemic equilibrium $P_4$ are stables (case~\eqref{thm_stability_J_3iv}). 

From the epidemiological point of view, we would like to avoid the case of stability of an endemic equilibrium. \textbf{Therefore, to avoid the stability of $P_4$, we must ensure that $\widetilde{a}_0 < \widetilde{a}_1/k$, that is, besides  $a_c < a_S$, we need that $a_c < \tfrac{a_I}{k} + a_R\left(1-\tfrac{1}{k}\right)$}. 

This is an ideal scenario that can be interpreted as disease eradication in the long-run.

\textbf{Scenario 2:} Consider now that $a_c > a_S$ (so $\widetilde{a}_0 > \widetilde{a}_2$), and still $R_0>1$.

In this case, the locally stable points will always be endemic: $P_4$ (cases~\eqref{thm_stability_J_2i} and~\eqref{thm_stability_J_3ii}) or $P_5$ (case~\eqref{thm_stability_J_3i}). Note however that, if $P_5$ exists ($R_p < R_0$), this equilibrium will represent a better situation than $P_4$, since the proportion of infected in $P_5$ will be lower than in $P_4$. Although $R_0$ does not depend on the payoff parameters, $R_p$ does, therefore \textbf{in order to obtain a lower proportion of infected, we must seek strategies such that the payoff parameters imply $R_p < R_0$. Furthermore, it is not enough that $P_5$ exists, we want $P_5$ to be stable. Then, in addition to $\widetilde{a}_0 > \widetilde{a}_2$ and $R_p < R_0$, we must also be sure that $\widetilde{a}_0 < \widetilde{a}_1/k$}. 

Note also that the components of $P_5$ depend on the value of $R_p$ and if $R_p$ goes to 1, the proportion of infected persons predicted by this equilibrium decrease. Given an infectious disease with $R_0> 1$, whereas it is not possible to change the inequality $\widetilde{a}_0 > \widetilde{a}_2$, it is possible to decrease the  number of infected people ensuring that $\widetilde{a}_0$ be less than $\widetilde{a}_1/k$ (so $P_5$ is stable) and as close as possible to $\widetilde{a}_2$. 

In this scenario, it is possible to quantify precisely the percentage of reduction on the infected population, produced by changes in the payoff parameters, as described in the following proposition.

\begin{proposition}\label{a0_reduction} 
Consider system~\eqref{model3} and assume that $R_0> 1$ and $\widetilde{a}_2<\widetilde{a}_0<\widetilde{a}_1/k$. A reduction of $p\%$  in $\widetilde{a}_0$ produce a reduction of $\left(\tfrac{\widetilde{a}_0}{\widetilde{a}_1-k\widetilde{a}_2}\right) p$ percentage points in the infected population on the endemic equilibrium state $P_5$ and a relative reduction of $\tfrac{\widetilde{a}_0}{\widetilde{a}_0 - \widetilde{a}_2}p\%$. 
\end{proposition}
\begin{proof}

We can compute the percentage point reduction by computing the difference between the old value of the proportion of the infected population at the equilibrium point $P_5$ (denoted by $\bar{I}$)  and the new value (denoted by $\widehat{I}$) obtained after the reduction on $a_0$.  Note that 
$$ \bar{I}=\tfrac{1}{k}\left(1-\tfrac{1}{R_p}\right)=\tfrac{1}{k}\left(1-\tfrac{\widetilde{a}_1 - k\widetilde{a}_0}{\widetilde{a}_1-k\widetilde{a}_2}\right)=\tfrac{\widetilde{a}_0-\widetilde{a}_2}{\widetilde{a}_1 - k\widetilde{a}_2},$$
so we have that
\begin{equation*}
\begin{split}
\bar{I} - \widehat{I} &=\frac{\widetilde{a}_0-\widetilde{a}_2}{\widetilde{a}_1 - k\widetilde{a}_2}-\frac{(1-\frac{p}{100})\widetilde{a}_0-\widetilde{a}_2}{\widetilde{a}_1 - k\widetilde{a}_2}\\				
&=\frac{\widetilde{a}_0}{\widetilde{a}_1 - k\widetilde{a}_2}\frac{p}{100}.
\end{split}
\end{equation*}
Therefore, the reduction of  $p\%$ in $\widetilde{a}_0$ can be interpreted as a reduction, in the long-term, of  
$\left(\tfrac{\widetilde{a}_0}{\widetilde{a}_1-k\widetilde{a}_2}\right) p$ percentage points in the proportion of infected population.

The corresponding relative reduction can be obtained as
\begin{equation*}
\begin{split}
\frac{(\bar{I} - \widehat{I})100}{\bar{I}} &= \frac{\widetilde{a}_0}{\widetilde{a}_1 - k\widetilde{a}_2}\frac{\widetilde{a}_1 - k\widetilde{a}_2}{\widetilde{a}_0 - \widetilde{a}_2}p\\
&=\frac{\widetilde{a}_0}{\widetilde{a}_0 - \widetilde{a}_2}p.
\end{split}
\end{equation*}
\noindent
So, a reduction of $p\%$ in $\widetilde{a}_0$ can be interpreted as a reduction, in the long-term, of $\tfrac{\widetilde{a}_0}{\widetilde{a}_0 - \widetilde{a}_2}p\%$ in the proportion of infected population. 
\end{proof}

\textbf{Example} Recent measles outbreaks have been associated with a lack of effective vaccination, mainly due to misinformation on the inherent risks of vaccines~\cite{WHO}. In terms of the model proposed in this paper, erroneously high valuations on vaccination risk could be interpreted as a high value for $a_c$ or equivalently, a high value for $\widetilde{a}_0$. In this context, it is relevant to ask how much $\widetilde{a}_0$ must be reduced to obtain, for example, a reduction of 1 percentage point on the infected population in the long-term. Under conditions on Proposition~\ref{a0_reduction}, this desired one percentage point reduction can be obtained by a reduction of $\left(\tfrac{\widetilde{a}_1-k\widetilde{a}_2}{\widetilde{a}_0}\right)\%$ in $\widetilde{a}_0$.

To obtain useful insights from last expression, besides considering the limitations and partial validity of using the proposed model for this specific disease, one should also be able to have estimation of $k$, $\widetilde{a}_0, \widetilde{a}_1$, and $\widetilde{a}_2$. These last parameters were just introduced in the present paper and as such, there are not estimations available yet.

For illustration purposes, we present in Figure~\ref{fig_a0} a heat map for $p$, the percentage reduction on $a_0$, depending on the values of $\widetilde{a}_1$ and $\widetilde{a}_2$, that would be necessary to obtain a 1 percentage point reduction on the infected population in the long-term, considering the value of $a_0$ as a normalized quantity equals to 1 and an estimated\footnote{An estimation of $k$ could be obtained from the equality $R_0=\tfrac{\beta}{\mu+\gamma}=\tfrac{\beta}{\mu k}$ in~\eqref{rescaling} so $k=\tfrac{\beta}{R_0 \mu}$. For measles, $R_0$ is commonly considered between 12-18, and in this example we consider it equals to 18. As discussed in~\cite{guerra2017lancet}, this estimation may not be adequate for all kinds of populations.  The risk of transmission of an infectious disease is closely related to the infection rate $\beta$ and we consider the worst-case scenario where both parameters are equals. For measles, we consider this value equal to 
90\%~\cite{perry2004clinical}
The constant $\mu$ can be estimated as the inverse of the mean life expectancy and we are considering $\mu=\tfrac{1}{76}$} value of $k$ equals to 3.8. 

From this estimations, we have for example that, if in comparison with $\widetilde{a}_0$, $\widetilde{a}_1$ is 10 times greater and $\widetilde{a}_2$ is a half,  then, to obtain a 1 percentage point reduction on the infected population in the long-term it is necessary at least a reduction of 8.1\% on $\widetilde{a}_0$.

\begin{figure}[ht!]
\centering
\includegraphics[width=0.5\textwidth]{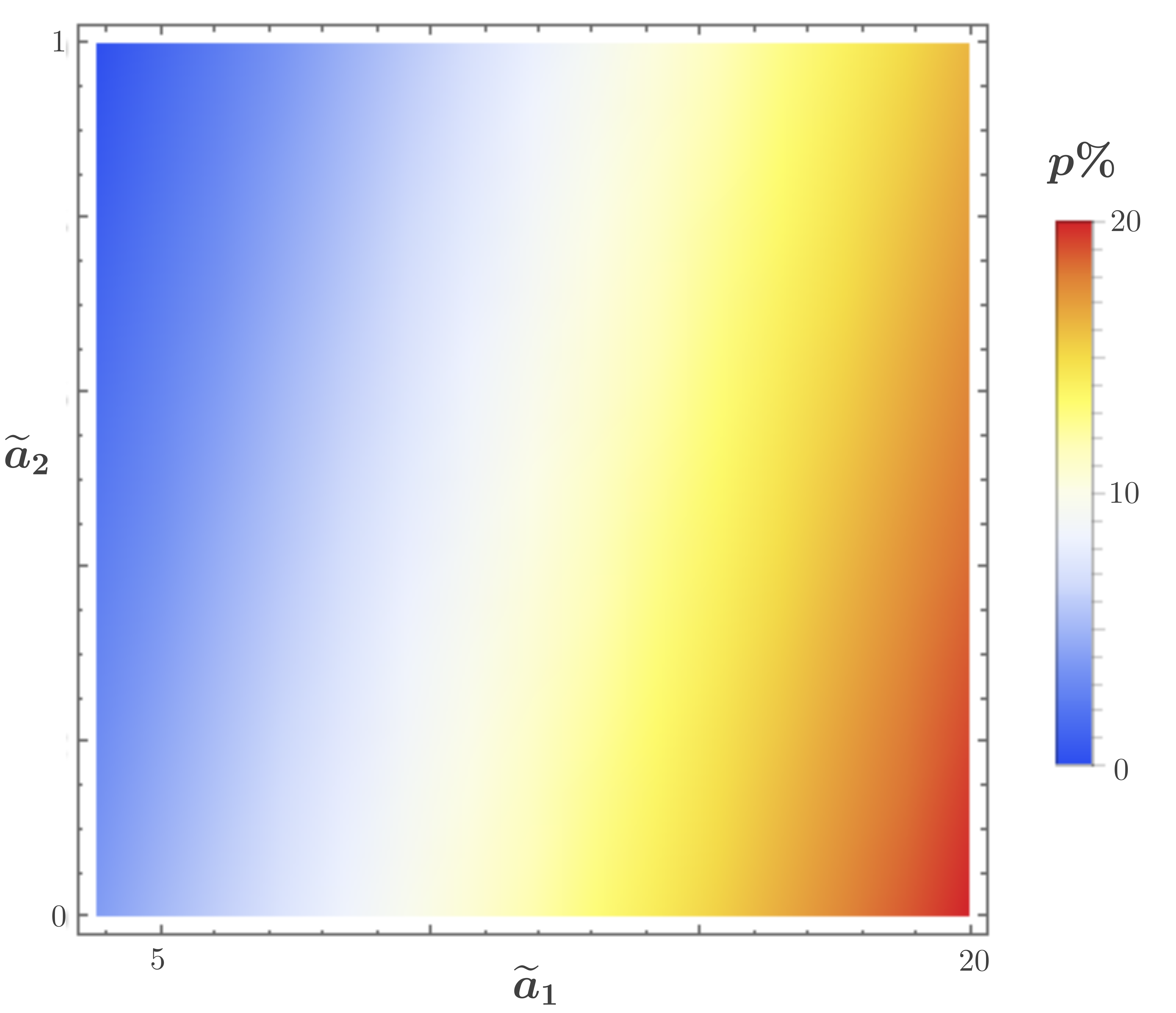}\label{fig_a0}
\caption{\small {\bf Heat map for $p$.}  Necessary percentage reduction on $a_0$ ($p$), depending on the values of $\widetilde{a}_1$ and $\widetilde{a}_2$, to obtain a reduction of 1 percentage point on the measles infected population in the long-term}
\end{figure}

\section{Final Comments}
\label{Sec:5}

The main contribution of this paper is the introduction of a mathematical model to analyze the interplay between infectious disease and anti-infection behavior adoption across the population. We focused on equilibrium states (Lemma~\ref{lemma_equilibriums}) and showed the appearance of remarkable characteristics in the context of epidemiological models (Theorem~\ref{thm_stability_J}), such as the coexistence of two locally stable endemic equilibria, the coexistence of a locally stable endemic and a disease-free equilibrium, and even the possibility of a stable continuum of endemic equilibrium points (Theorem~\ref{thm_stability_P6}).  We determine also the effects of some payoff reduction on the infected population size in an endemic equilibrium (Proposition~\ref{a0_reduction}). The expressions obtained in Proposition~\ref{a0_reduction} could be used as a threshold to estimate costs/payoff policies leading to long-term effective control of an epidemic infection through anti-infection behavior. Note that the relative reduction index obtained, depends only on the payoff parameters and not on the epidemiological parameters of the disease, therefore, it would be necessary to develop methods to estimate these payoff parameters for specific diseases and model validation from real data.

Other directions may be followed after the present work, to achieve real practical applicability of this kind of behavioral epidemiology modeling.  It would be reasonable to consider that the payoff function for the behavioral variable $x$, for the epidemiological variables, may not be linear. Also, it would be reasonable to consider that the payoff parameters are not necessarily constants and may vary on time. Different ways to model the variation and the effects of the behavioral variable $x$ can also be considered. Other models different from SIR can be suitable for specific situations, including  models considering delay differential equations to incorporate delayed effects/variations on the behavior adoption rate. We consider that the results obtained in the present work open valuable paths of research.

\appendix
\section{Proof of Theorem~\ref{thm_stability_J}}
\label{Appendix}

In this Appendix, we present the proof of Theorem~\ref{thm_stability_J}  based on the  Jacobian matrix  and characteristic polynomial \eqref{eq_polynomial}.

As mentioned before, any equilibrium point has at least one eigenvalue  $\lambda_1=-1$, and the other  eigenvalues can be studied  by analyzing the equation $q(\lambda) = 0$ for $P_1$, $P_2$, $P_4$, and $P_5$. This is described as follows.

\subsection*{Case: $P_1 = (1,0,0,0)$}
In this case, we have 
$$
q(\lambda) = \left|\begin{matrix} 
 -1 -\lambda& -k R_0  & 0 \\
 0& k R_0 -k -\lambda& 0 \\
 0 & 0 & -\widetilde{a}_0+\widetilde{a}_2 -\lambda   \\
\end{matrix}\right|.
$$
Thence, the additional eigenvalues are 
$$\lambda_2 = -1,  \quad \lambda_3 = k (R_0 -1)\quad \text{ and }\quad \lambda_4 = \widetilde{a}_2-\widetilde{a}_0.$$
\noindent
Therefore, if $R_0 < 1$ and $ \widetilde{a}_2<\widetilde{a}_0$, then all eigenvalues will be negative and, consequently, $P_1$ is locally asymptotically stable. If $R_0 > 1$ or $ \widetilde{a}_2>\widetilde{a}_0$, then $P_1$ is not stable.

\subsection*{Case: $P_2=(1,0,0,1)$}
In this case, we obtain
$$
q(\lambda)  = 
\left|\begin{matrix}
 -1 -\lambda& 0 & 0 \\
 0 & -k -\lambda& 0 \\
 0 & 0 & \widetilde{a}_0-\widetilde{a}_2 -\lambda\\
\end{matrix}
\right|.
$$
Thence, the additional eigenvalues are 
$$
\lambda_2 = -1,\quad \lambda_3 = -k, \quad \text{ and }\quad \lambda_4 = \widetilde{a}_0-\widetilde{a}_2.
$$
Therefore, it is sufficient that  $\widetilde{a}_0<\widetilde{a}_2$ for $P_2$ to be locally asymptotically stable. 
If $ \widetilde{a}_0>\widetilde{a}_2$, then $P_2$ is not stable.

\subsection*{Case: $P_4 = \left(\tfrac{1}{R_0},\tfrac{1}{k}\left(1-\tfrac{1}{R_0}\right),\left(1-\tfrac{1}{k}\right)\left(1-\tfrac{1}{R_0}\right),0\right)$}
In this case, we have that
$$
q(\lambda)\!=\!\left|
\begin{matrix}
 -\left(1-\frac{1}{R_0}\right) R_0-1-\lambda & -k & 1-\frac{1}{R_0} \\
 \left(1-\frac{1}{R_0}\right) R_0 & -\lambda & \frac{1}{R_0}-1 \\
 0 & 0 & -\widetilde{a}_0+\frac{\widetilde{a}_1 \left(1-\frac{1}{R_0}\right)}{k}+\frac{\widetilde{a}_2}{R_0}-\lambda\\
\end{matrix}\right|\!=\!q_1(\lambda)q_2(\lambda),
$$
where 
\begin{equation*}
\begin{split}
q_1(\lambda) & = \frac{R_0(\widetilde{a}_1 - \widetilde{a}_0 k)-\widetilde{a}_1 + \widetilde{a}_2 k }{k R_0}-\lambda \text{ and }\\
q_2(\lambda) &= \lambda^2+\lambda\left(\left( 1-\frac{1}{R_0}\right) R_0+1\right) +k\left(1-\frac{1}{R_0}\right) R_0.
\end{split}
\end{equation*}
Thence, the additional eigenvalues are $\lambda_2=\tfrac{R_0(\widetilde{a}_1 - \widetilde{a}_0 k)-\widetilde{a}_1 + \widetilde{a}_2 k }{k R_0}$ (the root of $q_1$) and the roots of the quadratic polynomial $q_2$. 
If $R_0> 1$, then the coefficients of $q_2$ are all positives and therefore from the Routh–Hurwitz criterion, we conclude that  eigenvalues associated with this polynomial must have negative real part. Note also that in this case  $\lambda_2 < 0$, if and only if $ R_0(\widetilde{a}_1 - \widetilde{a}_0 k)-\widetilde{a}_1 + \widetilde{a}_2 k < 0 $ or, equivalently, $R_0 (\widetilde{a}_1 - \widetilde{a}_0 k ) < \widetilde{a}_1 - \widetilde{a}_2 k$.

Therefore, $ \lambda_2 <0$ if and only if
\begin{itemize}
 \item $\widetilde{a}_1 - \widetilde{a}_0 k >0$ and $R_0<\tfrac{\widetilde{a}_1-\widetilde{a}_2k}{\widetilde{a}_1-\widetilde{a}_0k}=R_p$, or
 \item $\widetilde{a}_1 - \widetilde{a}_0 k <0$ and $R_0>\tfrac{\widetilde{a}_1-\widetilde{a}_2k}{\widetilde{a}_1-\widetilde{a}_0k}=R_p$, or
 \item $\widetilde{a}_1 - \widetilde{a}_0 k =0$ and $ \widetilde{a}_1 - \widetilde{a}_2 k >0 $.
\end{itemize}

If any of these conditions are satisfied, then $P_4$ is locally asymptotically stable.

\subsection*{Case: $P_5 = \left(\tfrac{1}{R_p},\tfrac{1}{k}\left(1-\tfrac{1}{R_p}\right),\left(1-\tfrac{1}{k}\right)\left(1-\tfrac{1}{R_p}\right), 1-\tfrac{R_p}{R_0}\right)$}

In this case we obtain
\begin{align*}
q(\lambda) 
&=
\left| 
\begin{matrix}
 -R_p -\lambda& -k & \frac{R_0}{R_p^2} (R_p-1) \\
 R_p-1 & -\lambda & -\frac{R_0}{R_p^2} (R_p-1) \\
 \frac{\widetilde{a}_2 (R_0-R_p) R_p}{R_0^2} & \frac{\widetilde{a}_1 (R_0-R_p) R_p}{R_0^2} &-\lambda
\end{matrix}
\right|=  -\lambda^3 - C_2 \lambda^2 - C_1\lambda  - C_0,
\end{align*}
\noindent
where
\begin{align*}
C_2 & = R_p \\
C_1 & = \frac{(R_p - 1)}{R_0R_p}\left((R_0-R_p)(\widetilde{a}_1-\widetilde{a}_2)+kR_0R_p\right), \text{ and } \\
C_0 & = \frac{(\widetilde{a}_1 - \widetilde{a}_2k)(R_0-R_p)(R_p-1)}{R_0R_p}.
\end{align*}

According to the Routh-Hurwitz criterion, the roots of $-q$ (also roots of $q$) will have the negative real part if, and only if,
$$
C_2 > 0, \quad C_0 > 0 \quad \text{and} \quad C_2C_1 - C_0 >0.
$$

If $1< R_p < R_0$, then we have immediately that $C_2  = R_p> 0$. 
Furthermore, in this case for $C_0>0$ it is necessary and sufficient that  
\begin{equation}\label{condP5a}
\widetilde{a}_1 - \widetilde{a}_2k >0.
\end{equation}

Additionally, note that 
\begin{align*}
C_2C_1 - C_0 
& =  \frac{(R_p - 1)(R_0-R_p)}{R_0}\left[(\widetilde{a}_0k-\widetilde{a}_2)+\frac{kR_0R_p}{R_0-R_p} \right]
\end{align*}
since  $(\widetilde{a}_1 - \widetilde{a}_0k)R_p = (\widetilde{a}_1 - \widetilde{a}_0k)\tfrac{\widetilde{a}_1 - \widetilde{a}_2k}{\widetilde{a}_1 - \widetilde{a}_0k} = \widetilde{a}_1 - \widetilde{a}_2k$.

Considering that $R_0>0$, $R_p>1$ and $R_0>R_p$, then
\begin{align} \label{condP5b}
C_2C_1 - C_0 > 0 
& \Leftrightarrow   (\widetilde{a}_0k-\widetilde{a}_2)+\frac{kR_0R_p}{R_0-R_p} >0 \nonumber\\
& \Leftrightarrow   \widetilde{a}_2 < k\left(\widetilde{a}_0+\frac{R_0R_p}{R_0-R_p}\right).
\end{align}

That is, $P_5$ is locally asymptotically stable, if and only if,~\eqref{condP5a} and~\eqref{condP5b} are satisfied  (s.t. the conditions of $P_5$ existence). 

Remember that the existence conditions for $P_5$ are
\begin{align}
\widetilde{a}_1\neq k\widetilde{a}_0\label{cond_existence_P5_1}\\
1 < R_p < R_0.  \label{cond_existence_P5_3}
\end{align}
Since $R_p = \tfrac{\widetilde{a}_1-k\widetilde{a}_2}{\widetilde{a}_1-k\widetilde{a}_0}$, to analyze inequality~\eqref{cond_existence_P5_3} we separate~\eqref{cond_existence_P5_1} in two cases:
\begin{itemize}
\item [Case 1:] $\widetilde{a}_1 - k\widetilde{a}_0 > 0$. \\
Multiplying~\eqref{cond_existence_P5_3} by $\widetilde{a}_1 - k\widetilde{a}_0$ we have
\begin{align*}
& \widetilde{a}_1-k\widetilde{a}_0 < \widetilde{a}_1-k\widetilde{a}_2 < R_0(\widetilde{a}_1-k\widetilde{a}_0)\\
\Rightarrow & -k\widetilde{a}_0 < -k\widetilde{a}_2 < R_0(\widetilde{a}_1-k\widetilde{a}_0) - \widetilde{a}_1\\
\Rightarrow & k\widetilde{a}_0 > k\widetilde{a}_2 > -R_0(\widetilde{a}_1-k\widetilde{a}_0) + \widetilde{a}_1.
\end{align*}

Joining the last inequality with the hypothesis considered in this case we have

\begin{equation}\label{cond_existence_P5_case1}
\widetilde{a}_1 > k\widetilde{a}_0 > k\widetilde{a}_2 > -R_0(\widetilde{a}_1-k\widetilde{a}_0) + \widetilde{a}_1.
\end{equation}

\item [Case 2:] $\widetilde{a}_1 - k\widetilde{a}_0 < 0$.\\
Analogously to the previous case, we will have

\begin{equation*}
\widetilde{a}_1 < k\widetilde{a}_0 < k\widetilde{a}_2 < -R_0(\widetilde{a}_1-k\widetilde{a}_0) + \widetilde{a}_1.
\end{equation*}
\end{itemize}

Note that, in order to ensure $P_5$ stability, is necessary that $\widetilde{a}_1 > \widetilde{a}_2k$ (condition~\eqref{condP5a}), which only occurs in~\eqref{cond_existence_P5_case1}.
However, if~\eqref{cond_existence_P5_case1} is satisfied, since $k>1$ and $R_0>R_p$ we have that  
$$
\widetilde{a}_2 < k\widetilde{a}_2 < k\widetilde{a}_0 < k\widetilde{a}_0 +  k\frac{R_0R_p}{R_0-R_p}. 
$$
That is, in the case~\eqref{cond_existence_P5_case1} the condition~\eqref{condP5b} is always satisfied.

Summarizing:
\begin{itemize}
\item If $\widetilde{a}_1 > k\widetilde{a}_0 > k\widetilde{a}_2 > -R_0(\widetilde{a}_1-k\widetilde{a}_0) + \widetilde{a}_1$, then $P_5$ exists and is locally asymptotically stable.
\item If $\widetilde{a}_1 < k\widetilde{a}_0 < k\widetilde{a}_2 < -R_0(\widetilde{a}_1-k\widetilde{a}_0) + \widetilde{a}_1$, then $P_5$ exists but is not stable. 
\item In other cases $P_5$ does not exist.
\end{itemize}

\bibliographystyle{elsarticle-num} 
\bibliography{references}

\end{document}